\documentclass[aps,prl,twocolumn,superscriptaddress]{revtex4-2}
\usepackage[utf8]{inputenc}
\usepackage{dsfont}
\usepackage{amsfonts}
\usepackage{amsmath}
\allowdisplaybreaks[4]        
\usepackage{amssymb}
\usepackage{euscript}     
\usepackage{braket}
\usepackage{starfont}
\usepackage{color,soul}         
\usepackage{tensor}        
\usepackage{amsthm}
\usepackage{graphicx}
\usepackage{slashed}
\usepackage{leftidx}
\usepackage{subfigure}
\usepackage{bbm}
\definecolor{outerspace}{rgb}{0.25, 0.29, 0.3}
\definecolor{scarlet}{rgb}{1.0, 0.13, 0.0}
\usepackage[header,title,page,titletoc]{appendix}  
\definecolor{princetonorange}{rgb}{1.0, 0.56, 0.0}
\definecolor{WildStrawberry}{rgb}{1.0, 0.26, 0.64}
\definecolor{rossocorsa}{rgb}{0.83, 0.0, 0.0}
\definecolor{navyblue}{rgb}{0.0, 0.0, 0.5}
\usepackage{float}
\usepackage[paper=letterpaper,margin=1in]{geometry}
\newtheorem{theorem}{Theorem}
\newtheorem{corollary}{Corollary}

\usepackage{mathtools}

\usepackage{ragged2e}



\usepackage{wasysym}
\usepackage{hyperref}
\hypersetup{
    colorlinks,
    citecolor=Peach,
    filecolor=Peach,
    linkcolor=MidnightBlue,
    urlcolor=MidnightBlue
}

\DeclareMathAlphabet{\pazocal}{OMS}{zplm}{m}{n}

\newcommand{\bea}{\begin{eqnarray}}

\newcommand{\eea}{\end{eqnarray}}
\newcommand{\ba}{\begin{eqnarray}}
\newcommand{\ea}{\end{eqnarray}}

\newcommand{\be}{\begin{equation}}
\newcommand{\ee}{\end{equation} }
\newcommand{\beqa}{\begin{eqnarray}}
\newcommand{\eeqa}{\end{eqnarray}}
\newcommand{\beqar}{\begin{eqnarray*}}
\newcommand{\eeqar}{\end{eqnarray*}}

\makeatletter
\newcommand{\dal}{\mathop{\mathpalette\dal@\relax}}
\newcommand{\dal@}[2]{%
  \begingroup
  \sbox\z@{$\m@th#1\square$}%
  \dimen0=\fontdimen8
    \ifx#1\displaystyle\textfont\else
    \ifx#1\textstyle\textfont\else
    \ifx#1\scriptstyle\scriptfont\else
    \scriptscriptfont\fi\fi\fi3
  \makebox[\wd\z@]{%
    \hbox to \ht\z@{%
      \vrule width \dimen0
      \kern-\dimen0
      \vbox to \ht\z@{
        \hrule height \dimen0 width \ht\z@
        \vss
        \hrule height 2\dimen0
      }%
      \kern-2.5\dimen0
      \vrule width 2.5\dimen0
    }%
  }%
  \endgroup
}
\makeatother

\usepackage[dvipsnames]{xcolor}
\usepackage{wasysym}
\usepackage{hyperref}

\hypersetup{
    colorlinks=true,
    citecolor=Peach,
    filecolor=Peach,
    linkcolor=MidnightBlue,
    urlcolor=MidnightBlue
}

\begin{document}

\title{All theories of gravity with second-order traced field equations}

\author{\'Angel J. Murcia}
\email{ajm386@cam.ac.uk}
\affiliation{Department of Applied Mathematics and Theoretical Physics,
University of Cambridge, \\ Wilberforce Road, Cambridge CB3 0WA, United Kingdom}


\begin{abstract}

\noindent 
We classify all diffeomorphism-invariant theories of gravity which possess traced field equations of motion of second order in derivatives. We also prove in the affirmative a conjecture regarding the form of theories constructed algebraically from the Riemann tensor with second-order traced field equations and classify all polynomial quasi-topological gravities with second-order traced field equations.
\end{abstract}
\maketitle

{\bf Introduction.} Among the most remarkable features of general relativity (GR), there is a fundamental underlying property which has enabled the outstanding progress in the theoretical, mathematical and numerical understanding of GR: the Einstein field equations are of second order in derivatives of the metric. This aspect is quite exceptional, as the Einstein-Hilbert action---from which
the Einstein field equations are derived---is also of
second order in derivatives. This constitutes a very peculiar case in which a second-order action, which would \emph{a priori} give rise to equations of motion of fourth order in derivatives, produces second-order field equations.

Therefore, an immediate question follows: which diffeomorphism-invariant theories of gravity constructed from contractions of Riemann curvature tensors and (possibly) covariant derivatives thereof give rise to equations of motion of second order in derivatives of the metric? This problem was solved more than 50 years ago by Lovelock \cite{Lovelock1,Lovelock2}, obtaining a result that depends on the space-time dimension $D$ that is being considered. For $D \leq 4$, the unique theory of gravity---up to total derivatives---is Einstein gravity supplemented with a cosmological constant. However, if one allows for higher dimensions $D \geq 5$, the answer becomes richer, as any linear combination of Euler densities would give rise to second-order field equations. For example, in $D=5,6$, the most general theory of gravity---up to total derivatives---with second-order equations is Einstein-Gauss-Bonnet with a cosmological constant.

In the following decades, additional research was pursued in the realm of higher-derivative theories of gravity, as these generically ensure renormalizability\footnote{At the cost of introducing ghosts and, consequently, losing unitarity.} of the theory \cite{Stelle:1976gc,Stelle:1977ry}, arise within an effective-field-theory approach to gravity \cite{Weinberg1995Quantum1} and appear as stringy corrections in the low-energy effective actions of string theory \cite{Gross:1986mw,Gross:1986iv,Grisaru:1986vi,Bergshoeff:1989de}. Although all theories of gravity with second-order equations had already been classified by Lovelock, these aspects motivated the search for more general classical theories of gravity which still retain some analytically intriguing properties. A particularly interesting class of theories entails the set of diffeomorphism-invariant theories of gravity with second-order traced field equations \cite{Oliva:2010eb,Oliva:2010zd}, which generalizes the class of Lovelock gravities. 

In this context, a natural problem corresponds to the classification of all theories of gravity whose traced field equations are of second order in derivatives of the metric. In this regard, some years ago a conjecture regarding all  algebraic (i.e., without explicit covariant derivatives) curvature invariants giving rise to second-order traced field equations was proposed
\cite{Oliva:2010zd,Oliva:2011xu}, which has ever since remained unproven. It is the purpose of this paper to classify all diffeomorphism-invariant theories of gravity with second-order traced field equations (without assuming the absence of covariant derivatives), proving as a corollary the aforementioned conjecture in the affirmative. 

The structure of the paper is as follows. First, we introduce some useful definitions and conventions. Second, we present our main theorems, which establish the most general form for any second-order traced field equations associated with a diffeomorphism-invariant theory of gravity and classify any such theory of gravity with second-order traced field equations. Next, we present the proof of both theorems. Afterwards, we present two corollaries in which we prove the conjecture presented in \cite{Oliva:2010zd,Oliva:2011xu} and classify all polynomial theories of gravity possessing second-order traced field equations and second-order equations of motion on top of spherical backgrounds. We conclude with a discussion of our results and some future directions.



{\bf Setup.} Let $(M,g)$ be a $D$-dimensional Lorentzian manifold. Let $R_{abc}{}^{d}$ be the Riemann curvature tensor associated with the Levi-Civita covariant derivative $\nabla$ of $g_{ab}$. For $D \geq 3$, we define the Weyl curvature tensor $W_{abc}{}^d$ and the Schouten tensor $S_{ab}$ as follows:
\begin{align}
\label{eq:Weyl}
W_{ab}^{cd}&=R_{ab}^{cd} - S_{[a}^{[c} \delta_{b]}^{d]}\,, \\
\label{eq:Schouten}
S_{ab}&=\frac{4}{(D-2)} \left[ R_{ab}-\frac{g_{ab}}{2(D-1)} R \right]\,,
\end{align}
where $R_{ab}$ stands for the Ricci tensor, $R$ denotes the Ricci scalar and antisymmetrizations of $p$ indices are defined with the usual weight $\frac{1}{p!}$. We also define:
\begin{align}
    \mathcal{R}_{a_1 \dots a_{2n}}^{a_1 \dots a_{2n}}&= R_{[a_1 a_2}^{a_1 a_2} \cdots R_{a_{2n-1}a_{2n}]}^{a_{2n-1}a_{2n}}\,, \\
     \mathcal{W}_{a_1 \dots a_{2n}}^{a_1 \dots a_{2n}}&= W_{[a_1 a_2}^{a_1 a_2} \cdots W_{a_{2n-1}a_{2n}]}^{a_{2n-1}a_{2n}}\,, \\
     \mathcal{S}_{a_1 \dots a_{n}}^{a_1 \dots a_{n}}&= S_{[a_1}^{a_1} \cdots S_{a_n]}^{a_{n}}\,.
\end{align}
We assume the functionals defining the diffeomorphism-invariant theories\footnote{No parity assumption is made on the theories.} of gravity to be at least twice continuously Fréchet differentiable, so that functional derivatives commute. The results of the paper are valid for every $D \geq 2$. For $D=2,3$, all results follow after setting to zero the Weyl tensor and any other possible scalar function thereof. For $D=2$, the proof of Theorem \ref{thm:1} needs to be adjusted as indicated in Footnote \cite{Note6}.


{\bf Main results.}  Let us present the two theorems to be proven in the manuscript. These are classification results for diffeomorphism-invariant theories of gravity whose traced field equations $\mathcal{E}$ contain at most second derivatives of the metric. On the one hand, Theorem \ref{thm:1} establishes the general form for such second-order\footnote{To facilitate the reading of the paper, by second order traced field equations we will always refer to traced field equations that are at most second order in derivatives.} traced field equations. On the other hand, Theorem \ref{thm:2} provides the expression of any diffeomorphism-invariant theory of gravity for which $\mathcal{E}$ is second order in derivatives of the metric.

\begin{theorem}
\label{thm:1}
Consider a general diffeomorphism-invariant theory of gravity
\begin{equation}
 I=\int \mathrm{d}^D x \sqrt{\vert g \vert }\,   \mathcal{L}\left (g^{ab}, R_{abcd}, \nabla_e \right)\,.
 \label{eq:genth}
\end{equation}
Assume its traced field equations $\mathcal{E}$ are of second order in derivatives of the metric. Then, for even space-time dimension\footnote{For $D=2,3$, we impose $\mathcal{F}=0$ in \eqref{eq:tracedeven} and \eqref{eq:tracedodd}.} $D=2k$ 
\begin{equation}
    \left. \mathcal{E}  \right \vert_{D=2k}=\alpha_0+ \sum_{n=1}^{k} \alpha_n \mathcal{R}_{a_1 \dots a_{2n}}^{a_1 \dots a_{2n}}+\mathcal{F}(W_{abcd})\,,
    \label{eq:tracedeven}
\end{equation}
while for odd dimension $D=2k-1$
\begin{align}
\notag
     \left. \mathcal{E}  \right \vert_{D=2k-1}&= \alpha_{k} \sum_{p=1}^{k} \gamma_{p}^{(k)} \mathcal{W}_{[a_1 \dots a_{2k-2p}}^{a_1 \dots a_{2k-2p}}\mathcal{S}_{a_{2k+1-2p} \dots a_{2k-p}]}^{a_{2k+1-2p} \dots a_{2k-p}}\\& +\alpha_0+ \sum_{n=1}^{k-1} \alpha_n \mathcal{R}_{a_1 \dots a_{2n}}^{a_1 \dots a_{2n}} +\mathcal{F}(W_{abcd})\,, 
     \label{eq:tracedodd}
\end{align}
where $\{\alpha_n\}_{n=0}^{[(D+1)/2]}$ are arbitrary couplings, $\mathcal{F}(W_{abcd})$ is an arbitrary scalar constructed out of Weyl curvature tensors and
\begin{equation}
    \gamma_p^{(k)}=\frac{k!(2k-p)! }{p(2k)! \left (k-p \right)! } \,.
    \label{eq:gammas}
\end{equation}
\end{theorem}

\begin{theorem}
\label{thm:2}
Up to total derivatives, a general diffeomorphism-invariant theory of gravity \eqref{eq:genth} in $D$ space-time dimensions will have  second-order traced field equations if and only if, for $D=2k$ 
\begin{align}
\notag
    \left. \mathcal{L}  \right \vert_{D=2k}&=-\frac{1}{k}\alpha_0+ \sum_{n=1}^{k-1} \frac{\alpha_n}{n-k} \mathcal{R}_{a_1\dots a_{2n}}^{a_1\dots a_{2n}}+\mathcal{L}_k^{\rm e}\\& +\mathcal{L}_{\rm W}(W_{abcd})+ \mathcal{L}_{\rm conf}\,,
    \label{eq:theoeven}
\end{align}
and, if $D=2k-1$, 
\begin{align}
\notag
    \left. \mathcal{L}  \right \vert_{D=2k-1}&=-\frac{2}{2k-1}\alpha_0+ \sum_{n=1}^{k-1} \frac{\alpha_n}{n-k+1/2} \mathcal{R}_{a_1\dots a_{2n}}^{a_1\dots a_{2n}}\\\notag  &+2 \alpha_{k}\sum_{p=1}^{k} \gamma_{p}^{(k)} \mathcal{W}_{[a_1 \dots a_{2k-2p}}^{a_1 \dots a_{2k-2p}}\mathcal{S}_{a_{2k+1-2p} \dots a_{2k-p}]}^{a_{2k+1-2p} \dots a_{2k-p}}\\&+\mathcal{L}_{\rm W}(W_{abcd})+\mathcal{L}_{\rm conf}\,,
    \label{eq:theoodd}
\end{align}
where $\{\alpha_n \}_{n=0}^{[(D+1)/2]}$ are arbitrary couplings, $\mathcal{L}_{\rm W}(W_{abcd})$ denotes any scalar function purely constructed from Weyl curvature tensors without covariant derivatives, $ \mathcal{L}_{\rm conf}$  (actually $\sqrt{\vert g \vert} \mathcal{L}_{\rm conf}$) corresponds to arbitrary terms invariant (up to total derivatives) under Weyl transformations of the metric\footnote{For the classification of local Weyl-invariant theories (up to total derivatives), we refer to \cite{Boulanger:2018rxo}.}, the coefficients $\gamma_p^{(k)}$ are presented in \eqref{eq:gammas} and $\mathcal{L}_k^{\rm e}$ is a non-local theory of gravity whose traced field equations are proportional to the $k$-th Euler density which may be presented as a local theory by introducing an auxiliary scalar field, see \eqref{eq:teoconfa}. Furthermore, the traced field equations of the theories $ \left. \mathcal{L}  \right \vert_{D=2k}$ and $ \left. \mathcal{L}  \right \vert_{D=2k-1}$ are given by \eqref{eq:tracedeven} and \eqref{eq:tracedodd}, respectively.
\end{theorem}

{\bf Proof of Theorem 1.}
Let $\delta_{\Omega_1} I$ be the first variation of the action with respect to variations of the form $\delta g^{ab}= \Omega_1 g^{ab}$, for nonzero arbitrary $\Omega_1 \in C^\infty(M)$. By definition,
\begin{equation}
    \delta_{\Omega_1} I= \int \mathrm{d}^D x \sqrt{\vert g \vert } \, \Omega_1 \mathcal{E} \,.
    \label{eq:conftrafo}
\end{equation}
where we remind that $\mathcal{E}$ stands for the trace of the gravitational field equations. Now consider the second variation $\delta_{\Omega_2} \delta_{\Omega_1} I$ of the action for another nonzero arbitrary smooth function $\Omega_2 \in C^\infty(M)$. Taking into account that variations commute, we have that
\begin{equation}
  [\delta_{\Omega_2}, \delta_{\Omega_1}] I=  \int \mathrm{d}^D x  A^{IJ}   \delta_{\Omega_J}\left (\sqrt{\vert g \vert }\, \Omega_I \, \mathcal{E} \right )=0\,,
  \label{eq:defconmu2}
\end{equation}
where $I,J=1,2$, $\Omega_I=\{\Omega_1,\Omega_2\}$ and $A^{IJ}$ is an antisymmetric matrix such that $A^{12}=-A^{21}=1$. By assumption, $\mathcal{E}$ is of second order in derivatives. The traced field equations $\mathcal{E}_{{\rm no} \nabla}$ of a diffeomorphism-invariant theory of gravity $\mathcal{L}_{{\rm no} \nabla}$ without covariant derivatives of the curvature are given by \cite{Rund1966Variational,Padmanabhan:2011ex}
\begin{equation}
\label{eq:traza2der}
\mathcal{E}_{{\rm no} \nabla}=P^{abcd} R_{abcd}-\frac{D \mathcal{L}_{{\rm no} \nabla}}{2} + 2  \nabla_c \nabla_d P^{cd}\,,
\end{equation}
where
\begin{align}
\label{eq:defP2der}
P^{abcd}&= \frac{\partial \mathcal{L}_{{\rm no} \nabla} }{\partial R_{abcd}}\,,\quad  P^{ac}=g_{bd} P^{abcd}\,.
\end{align}
Writing the Riemann curvature tensor in terms of the Weyl tensor $W_{abcd}$ and Schouten tensor $S_{ab}$---see \eqref{eq:Weyl} and \eqref{eq:Schouten}---, it is possible to show\footnote{In $D=2$, the corresponding formula is $P^{bd}=\dfrac{1}{2} \dfrac{\partial \mathcal{L}_{{\rm no}\nabla}}{\partial R} g^{bd}$. Also, in \eqref{eq:traced2der}, we would need to write $Q^{bd}=\dfrac{1}{2} \dfrac{\partial \mathcal{E}}{\partial R} g^{bd}$.} that
\begin{equation}
P^{bd}=g_{ac}\frac{\partial \mathcal{L}_{{\rm no} \nabla}}{\partial R_{abcd}}=\frac{\partial \mathcal{L}_{{\rm no} \nabla}}{\partial S_{bd}}\,.
\label{eq:formP}
\end{equation}
Using these results, after integration by parts one may rewrite \eqref{eq:defconmu2} as 
\begin{align}
  [\delta_{\Omega_1}, \delta_{\Omega_2}] I= 2\int  \mathrm{d}^D x  & \sqrt{\vert g \vert } \, A^{IJ} \Omega_I \nabla_c \Omega_J \nabla_ d Q^{cd}=0\,,
  \label{eq:defconmu}
\end{align}
where 
\begin{align}
\label{eq:traced2der}
Q^{abcd}&= \frac{\partial \mathcal{E} }{\partial R_{abcd}}\,,\quad  Q^{ac}=g_{bd} Q^{abcd}=\frac{\partial \mathcal{E}}{\partial S_{ac}}\,.
\end{align}
As $\Omega_1, \Omega_2$ are arbitrary, we conclude that
\begin{equation}
    \nabla_d Q^{cd}=0\,.
\end{equation}
All symmetric and divergenceless rank-two tensors which are concomitants of the metric and the curvature tensor were shown to be of the form \cite{Lovelock1,Lovelock2}:
\begin{equation}
    Q_{c}^d=\frac{\beta_1}{2}\delta_{c}^{d}+ \sum_{n=1}^{[(D-1)/2]} \frac{\beta_{n+1}}{2} \delta_{[c}^d \mathcal{R}_{a_1 \dots a_{2n}]}^{a_1 \dots a_{2n}}\,,
    \label{eq:Pres}
\end{equation}
where $\{\beta_n\}_{n=1}^{[(D+1)/2]}$ are arbitrary couplings and $[(D+1)/2]$ stands for the integer part of $(D+1)/2$. 
Plugging \eqref{eq:Pres} into \eqref{eq:traced2der}, we may integrate with respect to $S_{bd}$ to obtain the final result for the traced field equations. For even dimensions $D=2k$,
\begin{align}
\left. \mathcal{E} \right \vert_{D=2k}&=\alpha_0+\sum_{n=1}^{k}  \beta_n   \mathcal{A}_n 
 + \mathcal{F}(W_{abcd}) \,, \\
 & \hspace{-0.85cm }\mathcal{A}_n=\sum_{p=1}^n  \kappa_{n,p}  \mathcal{W}_{[a_1\dots a_{2n-2p}}^{a_1\dots a_{2n-2p}} \mathcal{S}_{a_{2n-2p+1} \dots a_{2n-p}]}^{a_{2n-2p+1} \dots a_{2n-p}}\,.
 \label{eq:tracefieldeqseven}
\end{align}
where $\alpha_0$ is an arbitrary constant, $\mathcal{F}(W_{abcd})$ is an undetermined scalar constructed entirely from the Weyl curvature tensor and
\begin{align}
\kappa_{n,p}=\binom{n}{p}\frac{\prod_{q=1}^{p-1} (D-2n+q+1)}{\prod_{q=0}^{p-1} (2n-q)}\,.
\end{align}
For $D=2k$ and $k\geq n$, taking into account that
\begin{equation}
 \left ( \mathcal{R}_{a_1 \dots a_{2n}}^{a_1 \dots a_{2n}} -\mathcal{W}_{a_1 \dots a_{2n}}^{a_1 \dots a_{2n}} \right)=(2k-2n+1)\mathcal{A}_n\,,  
\end{equation}
one may redefine $\mathcal{F}(W_{abcd})$ by adding  invariants $\mathcal{W}_{a_1 \dots a_{2n}}^{a_1 \dots a_{2n}}$ and get \eqref{eq:tracedeven} by setting $\beta_n=(2k-2n+1)\alpha_n$ for $k \geq n$.

On the other hand, for odd dimensions $D=2k-1$:
\begin{align}
 \label{eq:tracefieldeqsodd}
\left. \mathcal{E} \right \vert_{D=2k-1}&=\alpha_0+\sum_{n=1}^{k-1}  \beta_n   \mathcal{A}_n 
 + \mathcal{F}(W_{abcd}) \\ \notag &+\beta_{k} \sum_{p=1}^{k} \gamma_{p}^{(k)} \mathcal{W}_{[a_1 \dots a_{2k-2p}}^{a_1 \dots a_{2k-2p}}\mathcal{S}_{a_{2k+1-2p} \dots a_{2k-p}]}^{a_{2k+1-2p} \dots a_{2k-p}}\,,
\end{align}
where $\alpha_0$ and $\mathcal{F}$ are unspecified and $\gamma_p^{(k)}=\kappa_{k,p}$ was defined in \eqref{eq:gammas}. Redefining appropriately $\mathcal{F}(W_{abcd})$, one may convert  the invariants $\mathcal{A}_n$ into $\mathcal{R}_{a_1 \dots a_{2n}}^{a_1\dots a_{2n}}$ for $k > n$. Setting $\beta_n=2(k-n)\alpha_n $ for $1 \leq n \leq k-1$  and $\beta_k=\alpha_k$,  we get \eqref{eq:tracedodd}  and we conclude.

{\bf Proof of Theorem 2.} Let us assume we have two different diffeomorphism-invariant theories of gravity $\mathcal{L}$ and $\mathcal{L}'$ which give rise to the same traced field equations of motion. Let $I$ and $I'$ be the corresponding action functionals and take $\{\delta_{\Omega} I$, $\delta_{\Omega} I'\}$ to be the variations with respect to $\delta g^{ab}=\Omega g^{ab}$. Direct computation shows that
\begin{equation}
    \delta_\Omega I-\delta_{\Omega} I'=0
\end{equation}
by virtue of \eqref{eq:conftrafo}. Hence $\sqrt{\vert g \vert}(\mathcal{L}-\mathcal{L}')$ is invariant (up to total derivatives) under Weyl transformations $\delta g^{ab}=\Omega g^{ab}$. Consequently, the rest of the proof entails the finding of a diffeomorphism-invariant theory producing the traced field equations \eqref{eq:tracedeven} and \eqref{eq:tracedodd} in even and odd space-time dimensions, respectively. 

Let us begin by deriving a theory $\mathcal{L}_{\rm W}$ whose traced field equations are of the form $\mathcal{F}(W_{abcd})$, where $\mathcal{F}$ is an arbitrary scalar function constructed from contractions of Weyl curvature tensors. We express $\mathcal{F}(W_{abcd})=\mathcal{F}\left (\mathcal{I}_1^{(2)},\mathcal{I}_i^{(n)} \right)$, where $\mathcal{I}_1^{(2)}=W_{abcd} W^{abcd}$ and\footnote{In $D=4$, one could also consider the odd-parity invariant $W_{abcd} \tilde{W}^{abcd}$, where $\tilde{W}^{abcd}$ is the dual Weyl tensor. This does not affect the argument and just amounts to adding an invariant $\mathcal{I}_2^{(2)}$ in the functional dependence of $\mathcal{L}_{\rm W}$ and changing corresponding \eqref{eq:pdeW} and \eqref{eq:lwform}. } $\left\{\left \{ \mathcal{I}_i^{(n)}\right \}_{i=1}^{g_n}\right \}_{n=3}^\infty$ collects all  invariants constructed solely from the Weyl curvature of every possible order $n \geq 3$ (without covariant derivatives), with $g_n$ being the number of such independent invariants at $n$-th order. Writing $\mathcal{L}_{\rm W}=\mathcal{L}_{\rm W}(\mathcal{I}_1^{(2)}, \mathcal{I}_i^{(n)})$, we then have 
\begin{equation}
2 \frac{\partial \mathcal{L}_{\rm W}}{\partial \mathcal{I}_1^{(2)}}\mathcal{I}_1^{(2)} +\sum_{n=3}^\infty \sum_{i=1}^{g_n} n \frac{\partial  \mathcal{L}_{\rm W}}{\partial \mathcal{I}_i^{(n)}}  \mathcal{I}_i^{(n)}-\frac{D}{2} \mathcal{L}_{\rm W}=\mathcal{F}\,. 
\label{eq:pdeW}
\end{equation}
This is a linear partial differential equation for $\mathcal{L}_{\rm W}$. A formal solution to \eqref{eq:pdeW} is given by
\begin{equation}
    \mathcal{L}_W=\frac{\left( \mathcal{I}_1^{(2)} \right)^{D/4}}{2}\int^{\mathcal{I}_1^{(2)}} \frac{\mathrm{d}s\, \mathcal{F}\left (s,s^{n/2}\tilde{\mathcal{I}}_i^{(n)} \right) }{ s^{D/4+1}}\,,
    \label{eq:lwform}
\end{equation}
where $\tilde{\mathcal{I}}_i^{(n)}=\mathcal{I}_i^{(n)} \left (\mathcal{I}_1^{(2)} \right)^{-n/2} $ and where we emphasize that each $\tilde{\mathcal{I}}_i^{(n)}$ in the argument of $\mathcal{F}$ is multiplied by the factor $s^{n/2}$ that depends on the order of each curvature invariant $\mathcal{I}_i^{(n)}$.
The theory \eqref{eq:lwform} might have a very unpleasant form---maybe featuring non-polynomial combinations of curvature invariants or being non-analytic around conformally flat spaces, for example---, but its traced field equations yield precisely $\mathcal{F}$. This proves the existence of theories whose traced field equations produce any desired $\mathcal{F}(W_{abcd})$.

Let us continue by examining now the case of even space-time dimensions, $D=2k$. Trivially, a theory of the form
\begin{equation}
\mathcal{L}_{n<k}^{\rm e}=-\frac{1}{k}\alpha_0+ \sum_{n=1}^{k-1} \frac{\alpha_n}{n-k} \mathcal{R}_{a_1\dots a_{2n}}^{a_1\dots a_{2n}}
    \label{eq:lnk}
\end{equation}
gives rise to the traced field equations \eqref{eq:tracedeven} with $\alpha_k=\mathcal{F}=0$. This follows from the fact that \eqref{eq:lnk} is a Lovelock theory, for which the equations of motion are well known \cite{Lovelock1,Lovelock2}. The remaining task is to find a theory producing traced field equations proportional to $\mathcal{R}_{a_1 \dots a_{2k}}^{a_1 \dots a_{2k}}$. This entails a more difficult endeavor, as there is no Lovelock theory producing such traced field equations in dimensions $D=2k$. However, this problem is intimately connected with the conformal anomaly in theories of gravity with matter in even dimensions, which is obtained by varying the quantum effective action (resulting from integrating out the matter fields) with respect to variations $\delta g^{ab}=\Omega g^{ab}$. Indeed, the anomaly contains a term of the form $\mathcal{R}_{a_1 \dots a_{2k}}^{a_1 \dots a_{2k}}$ (called \emph{type A anomaly} \cite{Deser:1993yx}), which motivated the search of diffeomorphism-invariant theories of gravity giving rise to such traced field equations. This problem has been studied in the literature and  non-local diffeomorphism-invariant theories of gravity in $D=2k$ dimensions whose traced field equations are proportional to $\mathcal{R}_{a_1 \dots a_{2k}}^{a_1 \dots a_{2k}}$ have been found \cite{Polyakov:1981rd,Polyakov:1987zb,Riegert:1984kt,Paci:2024xxz}. These theories involve the inverse operator of a Weyl-covariant deformation of the $k$-th order Laplacian \cite{Graham1992Conformally}. As such, it involves infinite derivatives and the explicit form is not very illuminating. Nonetheless, it is possible to introduce a scalar auxiliary field to obtain a classically equivalent local diffeomorphism-invariant action for the metric and a scalar field\footnote{Under Weyl rescalings $\delta g^{ab}=\Omega g^{ab}$, the scalar field is required to transform as $\delta \tau=-\frac{\Omega}{2}$.} $\tau$ \cite{Mazur:2001aa,Baume:2013ika}. This reads \cite{Baume:2013ika}:
\begin{align}
\notag
    \mathcal{L}_{k}^{\rm e}&= \sum_{m=0}^{k-1}\sum_{l=0}^{k-m} c_{kml} (\partial \tau)^{2(k-m-l)} \mathcal{R}_{[a_1 \dots a_{2m}}^{a_1 \dots a_{2m}} (\mathcal{D}^2\tau)_{b_1 \dots b_l]}^{b_1 \dots b_l} \\ & -2 \alpha_k   \tau\, \mathcal{R}_{a_1 \dots a_{2k}}^{a_1 \dots a_{2k}}\,,
    \label{eq:teoconfa}
\end{align}
where $\alpha_k$ is a free parameter,
\begin{align}
(\mathcal{D}^2\tau)_{b_1 \dots b_l}^{c_1 \dots c_l}=\nabla_{b_1} \partial^{c_1} \tau \cdots \nabla_{b_l} \partial^{c_l} \tau
\end{align}
and $c_{kml}=0$ if $m=k-1$, $l=1$ or, otherwise,
\begin{equation}
   c_{kml}=
  \alpha_k\frac{(2m+l)!(2k-2m-l-2)! k!}{m! (2k)! (k-m-l)!l! (-2)^{m-l-k-1}} \,.
\end{equation} 
The trace of $\mathcal{L}_k^{\rm e}$, after eliminating the scalar auxiliary field, yields $\alpha_k \mathcal{R}_{a_1 \dots a_{2k}}^{a_1 \dots a_{2k}}$. This concludes the finding of the most general theory of gravity in even dimensions with  second-order traced field equations. 

Let us now consider the case of odd space-time dimensions $D=2k-1$. On the one hand, following a similar argumentation as in the even case, the theory 
\begin{equation}
\mathcal{L}_{n<k}^{\rm o}=-\frac{2}{D}\alpha_0+ \sum_{n=1}^{k-1} \frac{\alpha_n}{n-D/2} \mathcal{R}_{a_1\dots a_{2n}}^{a_1\dots a_{2n}}
    \label{eq:lnk2}
\end{equation}
possesses  second-order traced field equations of the form \eqref{eq:tracedodd} with $\alpha_k=\mathcal{F}=0$. To find a theory whose traced field equations correspond to \eqref{eq:tracedodd} with $\alpha_n=0$ for $0 \leq n<k$ and $\mathcal{F}=0$, consider the theory
\begin{equation}
    \mathcal{L}_{\rm k}^{\rm o}=2 \alpha_{k}\sum_{p=1}^{k} \gamma_{p}^{(k)} \mathcal{W}_{[a_1 \dots a_{2k-2p}}^{a_1 \dots a_{2k-2p}}\mathcal{S}_{a_{2k+1-2p} \dots a_{2k-p}]}^{a_{2k+1-2p} \dots a_{2k-p}}\,,
    \label{eq:specialtheory}
\end{equation}
where the constants $\gamma_p^{(k)}$ were defined in \eqref{eq:gammas}. By the computations in the proof of Theorem \ref{thm:1}, 
\begin{equation}
    \frac{\partial \mathcal{L}_{\rm k}^{\rm o}}{\partial S_d^c}=\alpha_k\delta_{[c}^d \mathcal{R}_{a_1 \dots a_{2k-2}]}^{a_1 \dots a_{2k-2}}\,.
\end{equation}
As a result, by the differential Bianchi identity,
\begin{equation}
    \nabla_f \frac{\partial \mathcal{L}_{\rm k}^{\rm o}}{\partial S_{ef}}=0\,.
\end{equation}
By \eqref{eq:traza2der}, the traced field equations of  $\mathcal{L}_{\rm k}^{\rm o}$ are given by $\frac{1}{2}\mathcal{L}_{\rm k}^{\rm o}$. In this way, for any odd space-time dimension we have been able to construct explicit theories that produce all terms in \eqref{eq:tracedodd}. Hence we conclude. 

{\bf Two corollaries}. From Theorems \ref{thm:1} and \ref{thm:2}, we will derive two interesting corollaries. The first of them proves in the affirmative a conjecture presented in \cite{Oliva:2010zd,Oliva:2011xu} about all polynomial curvature invariants  without covariant derivatives of a given curvature order producing second-order traced field equations.
\begin{corollary}
\label{coro1}
Any polynomial curvature invariant of order $k$ without covariant derivatives whose traced field equations are second order will be:
    \begin{itemize}
        \item In dimensions $D \neq 2k-1$, a linear combination of the $k$-th order Euler density, curvature invariants constructed uniquely from the Weyl curvature tensor and a total derivative.
        \item In dimensions $D=2k-1$, a linear combination of curvature invariants constructed uniquely from the Weyl curvature tensor, a total derivative and the theory \eqref{eq:specialtheory}.
    \end{itemize}
\end{corollary}
The proof of this corollary is a direct consequence of Theorem \ref{thm:2}, by direct inspection of \eqref{eq:theoeven} and \eqref{eq:theoodd}. Also, note that non-topological invariants without covariant derivatives transforming covariantly under Weyl transformations are constructed entirely from Weyl curvature tensors \cite{alexakis2006decomposition,Fefferman:2007rka}.

Secondly, we will classify diffeomorphism-invariant theories of gravity with second-order traced field equations that belong to the so-called quasi-topological (QT) class. These are local\footnote{Non-local extensions of QT gravities, defined so as to preserve the spherical section of the standard local ones, have recently been identified \cite{Bueno:2026oyg}.} theories of gravity whose equations of motion become second order in derivatives of the metric on top of spherically symmetric backgrounds \cite{Oliva:2010eb,Quasi,Bueno:2019ltp,Bueno:2019ycr,Bueno:2022res,Bueno:2025qjk}. These have recently attracted considerable attention, as the addition of infinite towers of higher-curvature terms of the QT class gives rise, generically, to regular black holes which are formed through dynamical collapse of matter \cite{Bueno:2024dgm,Bueno:2024zsx,Bueno:2024eig,Bueno:2025gjg}. We will focus on \emph{polynomial} QT gravities, i.e., QT theories without covariant derivatives which are constructed polynomially\footnote{Note that there exist non-polynomial QT gravities \cite{Colleaux:2017ibe,Bueno:2025zaj,Borissova:2026krh,Colleaux:2026hat} which may be defined in $D=4$ for any curvature order, while the unique polynomial QT gravity in $D=4$ is Einstein gravity.} from contractions of Riemann curvature tensors:

\begin{corollary}
\label{coro2}
 Let $\mathcal{Z}=\mathcal{Z}(g^{ab},R_{abcd})$ be a polynomial QT gravity with second-order traced field equations. Then:
 \begin{itemize}
     \item For even space-time dimensions $D=2k$,
     \begin{align}
    \left. \mathcal{Z}  \right \vert_{D=2k}&=\beta_0+ \sum_{n=1}^{k} \beta_n \mathcal{R}_{a_1\dots a_{2n}}^{a_1\dots a_{2n}}+\mathcal{F}_{\rm t}(W_{abcd}) \,,
    \label{eq:theoqteven}
    \end{align}
    \item For odd space-time dimensions $D=2k-1$,
     \begin{align}
     \notag
    \left. \mathcal{Z}  \right \vert_{D=2k-1}&=\beta_0+ \sum_{n=1}^{k-1} \beta_n \mathcal{R}_{a_1\dots a_{2n}}^{a_1\dots a_{2n}}+\mathcal{F}_{\rm t}(W_{abcd})\\   \notag &\hspace{-0.5cm}+2 \beta_{k}\sum_{p=1}^{k} \gamma_{p}^{(k)} \mathcal{W}_{[a_1 \dots a_{2k-2p}}^{a_1 \dots a_{2k-2p}}\mathcal{S}_{a_{2k+1-2p} \dots a_{2k-p}]}^{a_{2k+1-2p} \dots a_{2k-p}}\\& \hspace{-0.5cm}+ \lambda_k W_{a_1 a_2}^{a_{2k-1}a_{2k}} W_{a_3 a_4}^{a_1 a_2}\cdots W_{a_{2k-1}a_{2k}}^{a_{2k-3}a_{2k-2}} \,,
    \label{eq:theoqtodd}
\end{align}
 \end{itemize}
 where $\lambda_k$ is defined for every $k \geq 3$ as
 \begin{equation}
     \lambda_k=
\frac{8 \beta_k (k-1)^2[2(k-1)^2-1](2k-5)!}
{(2k)!(1+(k-2)^{k-1}[(2k-3)^{k-1}+2(-1)^k])}\,,
 \end{equation}
$\{\beta_n\}_{n=0}^{[(D+1)/2]}$ are free couplings, the coefficients $\gamma_p^{(k)}$ are defined in \eqref{eq:gammas} and $\mathcal{F}_{\rm t}(W_{abcd})$ is an arbitrary polynomial QT gravity built solely from Weyl tensors that vanishes on spherically symmetric backgrounds.
\end{corollary}
The proof follows by taking into account Theorem \ref{thm:2}, the fact that any Lovelock gravity is a QT gravity and equations (55) and (56) in \cite{Oliva:2010eb}. Also, as any QT gravity with non-trivial dynamics on spherical symmetry must contain terms with explicit Ricci curvature tensors, see \cite{Bueno:2022res,Moreno:2023rfl}, we conclude that $\mathcal{F}_{\rm t}(W_{abcd})$ must vanish on spherically symmetric backgrounds.

Corollary \ref{coro2}  extends Theorem 4 of \cite{Bueno:2025qjk}, as it gives the precise expression for polynomial QT gravities with second-order traced field equations.





{\bf Discussion}. In this work, we have completed the classification of all diffeomorphism-invariant theories of gravity (up to total derivatives) which possess traced field equations of second order in derivatives of the metric. We have also proven a conjecture regarding theories without covariant derivatives of the curvature in Corollary \ref{coro1} and extended a previous theorem about QT gravities with second-order traced field equations in Corollary \ref{coro2}.


Regarding future directions, there are various avenues to explore. One of them entails the classification of diffeomorphism-invariant theories of gravity and a scalar field with second-order traced field equations, or considering the addition of other types of matter, such as a gauge field. On a different front, our results show that, in four dimensions, the unique non-trivial modification to the traced field equations from Einstein gravity arises by adding the non-local term \eqref{eq:teoconfa} for $k=2$ to the original action. Therefore, it would be interesting to examine further the deformation of four-dimensional Einstein gravity by such a non-local theory, which has already been studied at the perturbative level \cite{Mottola:2006ew,Fabris:2011qq,Carballo-Rubio:2025eqy}.  Finally, focusing on the class of QT theories with second-order traced field equations, we observe that these only exist up to curvature order $[(D+1)/2]$ (including the Euler densities for even $D$). It would be of interest to find a special subclass of QT gravities which may be defined at any curvature order for $D \geq 5$ (in $D=4$, one needs to resort to non-polynomial QT theories) and generalize in a suitable way the condition of second-order traced field equations. This would be particularly useful to break the degeneracy of QT theories at higher curvature orders by selecting an analytically preferred subclass of them with which to study, for example, perturbations beyond spherical symmetry \cite{Bueno:2024zsx}.


\emph{Acknowledgments.} We thank Pablo Bueno and Robie A. Hennigar for useful comments. {\'A}JM is supported by the European Union’s Horizon Europe research and innovation programme
under the Marie Sk\l{}odowska-Curie grant agreement No. 101202045 (RoBHin-ETG).








\nocite{Boulanger:2018rxo,Bueno:2026oyg,Colleaux:2017ibe,Bueno:2025zaj,Borissova:2026krh,Colleaux:2026hat}
\bibliographystyle{JHEP}

\bibliography{Gravities}

\noindent 


\end{document}

{\bf Quasi-topological gravities with second-order traced field equations}. For odd curvature order $2n-1$:
\begin{widetext}
\begin{equation}
    \mathcal{Z}_{(2n-1)}=\mathcal{L}^{(2n-1)}_{\rm OR}- \binom{D}{2(2n-1)} \frac{2^{2n} D^{2n-2} (n-1)(D(2n-2)-2n+1)\left (W_{abcd} W^{abcd} \right)^{n-2} W_{efgh} W^{ghlm} W_{lm}^{ef}}{(D-2)^{4(n-1)}(D^2-4D+3)^{n-1} (D^3-9D^2+26D-22)} 
\end{equation}
For even curvature order $2n$, we have that:
\begin{equation}
 \mathcal{Z}_{(2n)}=   \mathcal{L}^{(2n)}_{\rm OR}- \binom{D-2}{4n} \frac{2^{2n}((2n-1)(D-2)((2n-1)D-2n))\left (W_{abcd} W^{abcd} \right )^n}{(D-2)^{2n} (D^2-4D+3)^{n} }\,.
\end{equation}

COMPROBAR
\end{widetext}

Now, consider the following combination of curvature invariants:
\begin{widetext}
    \begin{equation}
  \mathcal{L}_{\rm OR}^{(n)}=\sum_{q=1}^{n} \binom{n}{q}\frac{(D-2n+q)!}{(D-2n)! \prod_{i=0}^{q-1} (2n-i)} W_{[a_1 a_2}^{a_1 a_2} \cdots W_{a_{2n-2q-1} a_{2n-2q}}^{a_{2n-2q-1} a_{2n-2q}} S_{a_{2n-2q+1}}^{a_{2n-2q+1}} \cdots S_{a_{2n-q}]}^{a_{2n-q}} 
  \label{eq:teoriaor}
\end{equation} 
We will assume $D \geq 2n-1$, so that $\mathcal{L}_{\rm OR}^{(n)}$ is non-trivial. If $D > 2n-1$, $\mathcal{L}_{\rm OR}^{(n)}$ can be written as in \cite{Oliva:2010zd}:
\begin{equation}
    \mathcal{L}_{\rm OR}^{(n)}=\frac{1}{(D-2n+1)} \left[R_{[a_1 a_2}^{a_1 a_2} \cdots R_{a_{2n-1} a_{2n}]}^{a_{2n-1} a_{2n}}- W_{[a_1 a_2}^{a_1 a_2} \cdots W_{a_{2n-1} a_{2n}]}^{a_{2n-1} a_{2n}}\right]
\end{equation}
\end{widetext}
Nevertheless, this expression is not well defined for $D=2n-1$ and one needs to resort in this case to \eqref{eq:teoriaor}.

\begin{equation}
   \left. W_{[a_1 a_2}^{a_1 a_2} \cdots W_{a_{2n-1}a_{2n}]}^{a_{2n-1}a_{2n}} \right \vert_{\rm SS}= \Omega^n \binom{D-2}{2n} ((n-1)(D-1)((n-1)D-n))\,.
\end{equation}